\documentclass[a4paper,UKenglish,cleveref,autoref]{lipics-v2019}

\newcommand{\sub}{\subseteq}
\newcommand{\la}{\langle}
\newcommand{\ra}{\rangle}
\newcommand{\A}{\mathcal{A}}
\newcommand{\Z}{\mathbb{Z}}

\newcommand{\Q}{\mathbb{Q}}

\newcommand{\G}{\mathcal{G}}
\newcommand{\h}{\mathcal{H}}
\newcommand{\s}{\mathcal{S}}
\newcommand{\sg}{\mathrm{sg}}

\newcommand{\GL}{\mathrm{GL}(2,\Z)}

\newcommand{\Heis}{\mathrm{H}}
\renewcommand{\phi}{\varphi}
\renewcommand{\epsilon}{\varepsilon}

\title{On Reachability Problems for Low-Dimensional Matrix Semigroups}

\author{Thomas Colcombet}{IRIF, CNRS, Universit\'{e} Paris Diderot, France}{thomas.colcombet@irif.fr}{https://orcid.org/0000-0001-6529-6963}{Supported by the
		European Research Council (ERC) under the European Union’s
		Horizon 2020 research and innovation programme
		(grant agreement No.670624),
		and by the DeLTA ANR project (ANR-16-CE40-0007).}

\author{Jo\"{e}l Ouaknine}{The Max Planck Institute for Software Systems, Germany \and Department of Computer Science, University of Oxford, United Kingdom}{joel@mpi-sws.org}{https://orcid.org/0000-0003-0031-9356}{Supported by ERC grant AVS-ISS (648701) and by DFG grant 389792660 as part of TRR~248 (see \url{https://perspicuous-computing.science}).}

\author{Pavel Semukhin}{Department of Computer Science, University of Oxford, United Kingdom}{pavel.semukhin@cs.ox.ac.uk}{https://orcid.org/0000-0002-7547-6391}{Supported by ERC grant AVS-ISS (648701).}

\author{James Worrell}{Department of Computer Science, University of Oxford, United Kingdom}{jbw@cs.ox.ac.uk}{https://orcid.org/0000-0001-8151-2443}{Supported by EPSRC Fellowship EP/N008197/1.}

\authorrunning{T. Colcombet, J. Ouaknine, P. Semukhin, and J. Worrell}

\Copyright{Thomas Colcombet, Jo\"{e}l Ouaknine, Pavel Semukhin and James Worrell}



\nolinenumbers 

\hideLIPIcs

\EventEditors{Christel Baier, Ioannis Chatzigiannakis, Paola Flocchini, and Stefano Leonardi}
\EventNoEds{4}
\EventLongTitle{46th International Colloquium on Automata, Languages, and Programming (ICALP 2019)}
\EventShortTitle{ICALP 2019}
\EventAcronym{ICALP}
\EventYear{2019}
\EventDate{July 9--12, 2019}
\EventLocation{Patras, Greece}
\EventLogo{eatcs}
\SeriesVolume{132}
\ArticleNo{39}

\begin{document}
\maketitle

\begin{abstract}
  We consider the Membership and the Half-Space Reachability problems
  for matrices in dimensions two and three. Our first
  main result is that the Membership Problem is decidable for
  finitely generated sub-semigroups of 
  the Heisenberg group over rational numbers.
  Furthermore, we prove two decidability results for the Half-Space
  Reachability Problem. Namely, we show that this problem is decidable
  for sub-semigroups of $\mathrm{GL}(2,\mathbb{Z})$ and of the Heisenberg
  group over rational numbers.
\end{abstract}

\section{Introduction}

The algorithmic theory of matrix groups and semigroups is a staple of
computational algebra~\cite{Beals99} with numerous applications to
automata theory and program
analysis~\cite{BlondelJKP05,CK2005,DerksenJK05,HrushovskiOP018,KL86,MandelS77}
and has been influential in developing the notion of interactive
proofs in complexity theory~\cite{Babai85}.

Two central decision problems on matrix semigroups are the
\emph{Membership} and \emph{Half-Space Reachability} (see,
e.g.,~\cite{BP2012}).  For the Membership Problem the input is a
finite set of generators $A_1,\ldots,A_k$ and a target matrix $A$,
with all matrices being square and of the same dimension.  The
question is whether $A$ lies in the semigroup generated by
$A_1,\ldots,A_k$.  We emphasize that we consider membership in
finitely generated \emph{sub-semigroups}, i.e., we seek to recover $A$
as a non-empty product of generators.  In a related \emph{subgroup}
membership problem one additionally allows to take inverses of
generators.  The subgroup membership can clearly be reduced to
the sub-semigroup membership and tends to be more tractable (e.g.,
the subgroup membership for polycyclic groups is well-known to
be decidable~\cite{Sims94}, and the subgroup membership for the modular group
$\mathrm{PSL}(2,\mathbb{Z})$ is in PTIME~\cite{Gurevich2007}).
For the Half-Space Reachability
Problem the target matrix is replaced by vectors $\boldsymbol u$,
$\boldsymbol v$ and a scalar $\lambda$, and the question is now
whether there exists a matrix $A$ in the semigroup generated by
$A_1,\ldots,A_k$ such that
$\boldsymbol u^\top A \boldsymbol v \geq \lambda$.  Geometrically the
question is whether the orbit of $\boldsymbol v$ under the action of
the semigroup reaches a certain half-space with normal
$\boldsymbol u$. Closely related to these problems are the Vector
Reachability and the Hyperplane Reachability\footnote{In the
  literature the Hyperplane Reachability Problem is also called the
  Scalar Reachability Problem.} problems, which ask whether there
exists a matrix $A$ in the semigroup generated by $A_1,\ldots,A_k$
such that $A \boldsymbol v = \boldsymbol u$ or such that
$\boldsymbol u^\top A \boldsymbol v = \lambda$, respectively.

Undecidability of the Membership Problem has long been known (indeed,
this was one of the earliest undecidability results---see
A.~Markov~\cite{Markov}).  Subsequently a number of positive
decidability results were obtained in the case of semigroups generated
by commuting matrices over infinite fields~\cite{Babai,KL86}.  More
recently, attention has focussed on integer matrices in dimension two.
A classical result of \cite{CK2005} shows decidability of the
Membership Problem for sub-semigroups of
$\mathrm{GL}(2,\mathbb{Z})$---the group of $2\times 2$ integer
matrices with integer inverses (equivalently, with determinants equal
to $\pm 1$). Moreover, the semigroup membership for the identity matrix was shown to be
NP-complete for $\mathrm{SL}(2,\mathbb{Z})$~\cite{BHP17}.
Furthermore, the Membership Problem is decidable for $2\times 2$
integer matrices with nonzero determinant~\cite{PS_SODA} and for
$2\times 2$ integer matrices with determinants equal to $0$ and $\pm
1$~\cite{PS17}. However it is still unknown whether the Membership
Problem is decidable for all $2\times 2$ integer matrices.

Going beyond dimension two, it has long been known that the Membership
Problem is undecidable for general $3\times 3$ integer
matrices~\cite{Pat70}.  However the status of the Membership Problem
for $\mathrm{GL}(3,\mathbb{Z})$ is currently an outstanding open
problem.  Related to this, it was shown in~\cite{KNP18} that for a
two-element alphabet $\Sigma$, the monoid $\Sigma^*\times \Sigma^*$
cannot be embedded in $\mathrm{GL}(3,\mathbb{Z})$.  This fact suggests
that undecidability proofs of the Membership Problem in other settings
(such as~\cite{Pat70}), which are based on encodings of the Post
Correspondence Problem, are unlikely to carry over to
$\mathrm{GL}(3,\mathbb{Z})$.  It is classical that the Membership
Problem for $\mathrm{GL}(4,\mathbb{Z})$ is
undecidable~\cite{Mih58,Identity,KNP18}; thus it can reasonably be said that
dimension three lies on the borderline between decidability and
undecidability.

Our first main result (Theorem \ref{thm:heis}) concerns the Membership
Problem for a simple subgroup of $\mathrm{GL}(3,\mathbb{Z})$: the
so-called \emph{Heisenberg group} $\mathrm{H}(3,\mathbb{Z})$, which
comprises upper triangular integer matrices with ones along the diagonal.
Since the Heisenberg group is polycyclic, the sub\emph{group}
membership problem is decidable~\cite{Sims94}.  It was moreover
recently shown in~\cite{KNP18} how to decide membership of the
identity matrix in finitely generated sub-semigroups of
$\mathrm{H}(3,\mathbb{Z})$.  Our main theorem strengthens this last
result to show decidability of the Membership Problem for
$\mathrm{H}(3,\mathbb{Z})$. In fact, like in \cite{KNP18}, our argument works
for Heisenberg groups of any dimension and even over the field of rational numbers, that is,
for $\mathrm{H}(n,\mathbb{Q})$.

Our proof relies on arguments developed in~\cite{KNP18} but contains several significant new
elements, including the use of linear programming, integer register automata, and matrix logarithms.
The following algebraic property of $\mathrm{H}(3,\mathbb{Z})$ is important for our construction: the
subgroup generated by commutators of matrices from a given subset $\mathcal{G}\subseteq
\mathrm{H}(3,\mathbb{Z})$ is isomorphic to a subgroup of $\mathbb{Z}$. Such property does not hold
for the direct product of two Heisenberg groups $\mathrm{H}(3,\mathbb{Z})^2$ or for the group of
$4\times 4$ upper unitriangular matrices $\mathrm{UT}(4,\mathbb{Z})$. This makes it challenging to
generalize our argument to show decidability of the Membership Problem for
$\mathrm{H}(3,\mathbb{Z})^2$, $\mathrm{UT}(4,\mathbb{Z})$ or other similar matrix groups.

In~\cite{KLZ15} a related
problem was studied, called the Knapsack Problem. Namely, it was
proved that the Knapsack Problem is decidable for
$\mathrm{H}(3,\mathbb{Z})$, that is, given matrices $A_1,\ldots,A_k$
and $A$ from $\mathrm{H}(3,\mathbb{Z})$ one can decide whether there
are non-negative integers $n_1,\ldots,n_k$ such that
$A_1^{n_1}\cdots A_k^{n_k}=A$.  Decidability of the Knapsack Problem
is shown by reduction to the problem of solving a single quadratic
equation in integer numbers (proved to be decidable
in~\cite{GS81,SG04}).  By contrast, our decision procedure for the
Membership Problem relies only on linear programming and integer
linear arithmetic.  As far as we can tell, there is no straightforward
reduction in either direction between the Membership and Knapsack
Problems for $\mathrm{H}(3,\mathbb{Z})$.

The Vector Reachability, Hyperlane Reachability, and Half-Space
Reachabilty Problems are all known to be undecidable in general
(see~\cite{CassaigneHHN14,HH07,HHH2007}).
The Vector and Hyperplane Reachability problems are known to be decidable for
$\mathrm{GL}(2,\mathbb{Z})$, as shown in~\cite{PS19}. 
For matrix semigroups with a single generator,
the Half-Space Reachability Problem is equivalent to the
\emph{Positivity Problem} for linear recurrence sequences: a
longstanding and apparently difficult open
problem~\cite{OW14:SODA,RS94}.  Our second main result is that the
Half-Space Reachability Problem is decidable for both
$\mathrm{GL}(2,\mathbb{Z})$ (Theorem \ref{thm:halfpl}) and
$\mathrm{H}(n,\mathbb{Q})$ (Theorem \ref{thm:heis_half}).  For
$\mathrm{GL}(2,\mathbb{Z})$ we build on automata-theoretic techniques
developed in~\cite{CK2005}, with the key insight being that the set of
matrices in $\mathrm{GL}(2,\mathbb{Z})$ with a positive value in a given entry
can be represented as a regular language over the
generators of $\mathrm{GL}(2,\mathbb{Z})$.  For
$\mathrm{H}(n,\mathbb{Q})$ we rely on a nontrivial
result about the nonnegativity of quadratic forms over the integers from~\cite{GS81,SG04} (related
to the result used in~\cite{KLZ15} to solve the Knapsack Problem).

\section{Preliminaries}
\paragraph*{The Heisenberg Group.}
We use notations $I_n$ and $0_n$ for the identity matrix and for the zero matrix
of size $n\times n$, respectively. For $n\geq 3$, the \emph{Heisenberg group} of
dimension $n$ is the group $\Heis(n,\mathbb{R})$ of $n\times n$ real matrices of
the form
\begin{gather}
  A = \begin{pmatrix} 1 & \boldsymbol a^\top & c \\ 0& I_{n-2} &\boldsymbol b\\ 0&0&1 \end{pmatrix},
\label{eq:heisenberg}
\end{gather}
where $\boldsymbol a,\boldsymbol b \in \mathbb{R}^{n-2}$, $c\in \mathbb{R}$.
For brevity, we will often denote a matrix $A$ as in~\eqref{eq:heisenberg} by
the triple $(\boldsymbol a,\boldsymbol b,c) \in \mathbb{R}^{n-2} \times
\mathbb{R}^{n-2} \times \mathbb{R}$.  It is easy to check that the product
operation is given by
\[
  (\boldsymbol a,\boldsymbol b,c)\cdot (\boldsymbol a',\boldsymbol b',c') =
  (\boldsymbol a + \boldsymbol a',\boldsymbol b + \boldsymbol b',c+c' +
  \boldsymbol a^\top\boldsymbol b')\,.
\]
We use $\psi$ to denote the group homomorphism
$\psi:\Heis(n,\mathbb{R})\rightarrow \mathbb{R}^{2n-4}$ given by
$\psi(\boldsymbol a,\boldsymbol b,c)=(\boldsymbol a,\boldsymbol b)$.

The Heisenberg group $\Heis(n,\mathbb{R})$ is a Lie group whose
corresponding Lie algebra $\mathfrak{h}(n,\mathbb{R})$ comprises the vector
space of $n\times n$ real matrices of the form
\begin{gather} B= \begin{pmatrix} 0 & \boldsymbol a^\top & c \\ 0& 0_{n-2} &\boldsymbol b\\ 0&0&0 \end{pmatrix}, 
\label{eq:lie-alg}
\end{gather}
where $\boldsymbol a,\boldsymbol b \in \mathbb{R}^{n-2}$ and $c\in
\mathbb{R}$, together with the binary \emph{Lie bracket} operation
$[A,B]:=AB-BA$ for $A,B \in \mathfrak{h}(n,\mathbb{R})$. Note that $[A,B]$
has only zero entries except for the $(1,n)$-entry. From this it is easy to
check that $[[A,B],C]=0_n$ for all $A,B,C \in \mathfrak{h}(n,\mathbb{R})$.

Given $A \in \Heis(n,\mathbb{R})$, as shown in~\eqref{eq:heisenberg}, we
define its logarithm $\log(A) \in\mathfrak{h}(n,\mathbb{R})$ to be
\[ \log(A) := (A-I) - \frac{(A-I)^2}{2} =
  \begin{pmatrix} 0 & \boldsymbol a^\top & c-\textstyle\frac{1}{2}\boldsymbol a^\top
      \boldsymbol b \\ 0& 0_{n-2} &\boldsymbol
    b\\ 0&0&0 \end{pmatrix} . \] Conversely, given $B \in
  \mathfrak{h}(n,\mathbb{R})$, as shown in~\eqref{eq:lie-alg}, we
  define its exponential $\exp(B) \in \Heis(n,\mathbb{R})$ to
  be $\exp(B) := I + B +\frac{B^2}{2} = (\boldsymbol a, \boldsymbol b,
  c+\textstyle\frac{1}{2}\boldsymbol a^\top \boldsymbol b)$. It is easy to
  verify that $\log$ and $\exp$ are mutually inverse and together induce a
  bijection between $\Heis(n,\mathbb{R})$ and $\mathfrak{h}(n,\mathbb{R})$.

The following is a specialisation to $ \Heis(n,\mathbb{R})$
of the Baker-Campbell-Hausdorff product formula (see~\cite[Chapter 5]{hall2015} for a details).
Given a sequence of matrices $B_1,\ldots,B_m \in \Heis(n,\mathbb{R})$, we have
\begin{gather}
\log(B_1 \cdots B_m) = \sum_{i=1}^m \log(B_i) + 
\frac{1}{2} \sum_{1\leq i<j\leq m} [\log(B_i),\log(B_j)] \, . 
\label{eq:baker-campbell-hausdorff}
\end{gather}

\paragraph*{Regular subsets of $\GL$.}
We will use the notation $\GL$ for the general linear group of $2\times 2$ integer matrices, that is,
$
\GL=\{M\in \Z^{2\times 2} : \det(M)=\pm 1\}
$.
A matrix is called \emph{singular} if its determinant is zero and \emph{nonsingular} otherwise.

We will use the following encoding of the matrices from $\GL$ by words
in alphabet $\Sigma=\{X,N,S,R\}$. First, we define a mapping
$\phi: \Sigma\to \GL$ as follows:
\[
\phi(X)=-I_2=\begin{pmatrix} -1 & 0\\ 0 & -1\end{pmatrix}\!,\ 
\phi(N)=\begin{pmatrix} 1 & 0\\ 0 & -1\end{pmatrix}\!,\ 
\phi(S)=\begin{pmatrix} 0 & -1\\ 1 & 0\end{pmatrix}\!,\ 
\phi(R)=\begin{pmatrix} 0 & -1\\ 1 & 1\end{pmatrix}\!.
\]
We can extend $\phi$ to a morphism $\phi: \Sigma^* \to \GL$ in a
natural way. It is a well-known fact that morphism $\phi$ is
surjective, that is, for every $M\in \GL$ there is a word
$w\in \Sigma^*$ such that $\phi(w)=M$.
This presentation is not unique because of identities such as
$\phi(\mathit{SS})=\phi(\mathit{RRR})=\phi(X)$. However, as explained
below, every matrix $M\in \GL$ is represented by a unique word in the
\emph{canonical} form.

In the following definition, for $n$ a positive integer and
$V\in \Sigma$, $V^n$ is the word consisting of $n$ copies of $V$,
while $V^0$ denotes the empty word.
\begin{definition}
A word $w\in \Sigma^*$ is called \emph{canonical} if it has the form
\[
w=N^\delta X^\gamma S^\beta R^{\alpha_1}SR^{\alpha_2}\cdots SR^{\alpha_n}S^\epsilon,
\]
where $\beta,\gamma, \delta, \epsilon\in \{0,1\}$ and $\alpha_i\in \{1,2\}$ for $i=1,\ldots,n$. In other words, $w$ is \emph{canonical} if it does not contain subwords $SS$ or $RRR$. Moreover, letter $N$ may appear only once in the first position, and letter $X$ may appear only once either in the first position or after $N$.
\end{definition}

The next proposition is a well-known fact.

\begin{proposition}[\cite{LS,MKS,PS_SODA,Ran}]\label{prop:can}
For every matrix $M\in \GL$, there is a unique canonical word $w$ such that $M=\phi(w)$.
\end{proposition}

\begin{definition}
A subset $\s\sub \GL$ is called \emph{regular} if there is a regular language $L\sub \Sigma^*$ such that $\s=\phi(L)$.
\end{definition}

\begin{definition}
Two words $w_1$ and $w_2$ from $\Sigma^*$ are \emph{equivalent}, denoted $w_1\sim w_2$, if $\phi(w_1)=\phi(w_2)$.
Two languages $L_1$ and $L_2$ in the alphabet $\Sigma$ are \emph{equivalent}, denoted $L_1\sim L_2$, if
\begin{enumerate}[(i)]
  \item for each $w_1\in L_1$, there exists $w_2\in L_2$ such that $w_1\sim w_2$, and
  \item for each $w_2\in L_2$, there exists $w_1\in L_1$ such that $w_2\sim w_1$.
\end{enumerate}
In other words, $L_1\sim L_2$ if and only if $\phi(L_1)=\phi(L_2)$.
Two finite automata $\A_1$ and $\A_2$ with alphabet $\Sigma$ are \emph{equivalent}, denoted $\A_1\sim \A_2$, if $L(\A_1)\sim L(\A_2)$.
\end{definition}

The following theorem is a crucial ingredient of our decidability results.

\begin{theorem}[\cite{PS_SODA}] \label{thm:canon}
  For any automaton $\A$ over the alphabet $\Sigma=\{X,N,S,R\}$, there exists an automaton $\mathrm{Can}(\A)$ such that $\mathrm{Can}(\A)$ is equivalent to~$\A$ and $\mathrm{Can}(\A)$ accepts only canonical words. Furthermore, $\mathrm{Can}(\A)$ can be constructed from $\A$ in polynomial time.
\end{theorem}

From this theorem we obtain the following corollary.

\begin{corollary} \label{thm:bool}
  Regular subsets of $\GL$ are effectively closed under Boolean operations.
  Namely, given two regular languages $L,L'\sub \Sigma^*$, we can algorithmically construct in polynomial time regular languages $L^\cup$, $L^\cap$ and $L^c$ such that
  \[
    \phi(L^\cup)=\phi(L)\cup \phi(L'),\quad \phi(L^\cap)=\phi(L)\cap \phi(L'),\quad\text{and}\quad \phi(L^c)=\GL\setminus \phi(L).
  \]
\end{corollary}

\begin{proof}
  Let $\A$ and $\A'$ be finite automata that recognise the languages $L$ and $L'$, respectively. Using Theorem \ref{thm:canon}, we can algorithmically construct the automata $\mathrm{Can}(\A)\sim \A$ and $\mathrm{Can}(\A')\sim \A'$ that accept only canonical words. Recall that, by Proposition \ref{prop:can}, the matrices from $\GL$ are in one-to-one correspondence with the regular set $L_\mathrm{Can}$ of all canonical words. Therefore, we can define $L^\cup$, $L^\cap$ and $L^c$ to be the following regular sets:
  \begin{align*}
    L^\cup &= L(\mathrm{Can}(\A)) \cup L(\mathrm{Can}(\A')),\quad L^\cap = L(\mathrm{Can}(\A)) \cap L(\mathrm{Can}(\A'))\quad \text{and}\\
    L^c &=L_\mathrm{Can} \setminus L(\mathrm{Can}(\A)).
  \end{align*}
\end{proof}

\paragraph*{Decision problems for matrix semigroups.}

If $\G$ is a finite collection of matrices, then $\la \G\ra$ denotes the semigroup generated by
$\G$, that is, $A\in \la \G\ra$ if and only if there are matrices $A_1,\dots,A_t\in \G$ such that
$A=A_1\cdots A_t$.

In this paper we will consider the following decision problems for matrix semigroups:

\begin{itemize}
  \item\textbf{The Membership Problem:} Given a finite collection of matrices $\G$ and a ``target''
    matrix $A$, decide whether $A$ belongs to $\la \G\ra$.

  \item\textbf{The Half-Space Reachability Problem:} Given a finite collection of matrices $\G$, two
    vectors $\boldsymbol{u},\boldsymbol{v}$ and a scalar $\lambda$, decide whether there exists a matrix
    $A\in \la \G\ra$ such that $\boldsymbol{u}^\top A\boldsymbol{v}\geq \lambda$.
    In other words, decide whether it is possible to reach the half-space $\h = \{\boldsymbol{x} :
    \boldsymbol{u}^\top \boldsymbol{x} \geq \lambda\}$ using matrices from $\G$ starting from an
    initial vector $\boldsymbol{v}$.
\end{itemize}

When we talk about \emph{the Membership Problem for $\GL$ or for the Heisenberg group
$\Heis(n,\Q)$}, we mean that $A$ and the matrices from $\G$ belong to $\GL$ or $\Heis(n,\Q)$,
respectively.
Similarly, in \emph{the Half-Space Reachability Problem for $\GL$ or
$\Heis(n,\Q)$} we assume that $\G$ is a finite subset of $\GL$ or $\Heis(n,\Q)$, respectively, and
furthermore we assume that the vectors $\boldsymbol{u},\boldsymbol{v}$ have rational coefficients and $\lambda$ is a
rational number.

\section{The Membership Problem for the Heisenberg Group}\label{sec:heis}

Let $\Heis(n,\mathbb{Z})$ and $\Heis(n,\mathbb{Q})$ be subgroups of $\Heis(n,\mathbb{R})$
comprising all matrices with integer and rational entries, respectively.
In this section we will prove our first main result.

\begin{theorem}\label{thm:heis}
The Membership Problem for $\Heis(n,\mathbb{Z})$ is decidable.
\end{theorem}

We first give an overview of our decision procedure.  Let
$\mathcal{G}=\{A_1,\ldots,A_k\}$ be a finite set of generators from
$\Heis(n,\mathbb{Z})$ and $A \in \Heis(n,\mathbb{Z})$ be a target matrix.
The idea is to partition the set of generators $\mathcal{G}$ into two
sets $\mathcal{G}_+$ and $\mathcal{G}_0$.  The definition of
$\mathcal{G}_+$ is such that there is a computable upper bound on the
number of occurrences of a matrix from $\mathcal{G}_+$ in any string
of generators whose product equals the target matrix $A$.  The
definition of $\mathcal{G}_0$ is such that the image of the semigroup
generated by $\mathcal{G}_0$ under the homomorphism $\psi$ is a
sub\emph{group} of $\mathbb{R}^{2n-4}$ (i.e., the image is closed under
inverses).  We then proceed by a case analysis according to whether or
not $\mathcal{G}_0$ is a commutative set of matrices.  If
$\mathcal{G}_0$ is commutative then the Membership Problem can be
reduced to solving a system of linear equations over non-negative
integer variables.  If $\mathcal{G}_0$ is not commutative then we
reduce the Membership Problem to a reachability query in an integer
register automaton.

\paragraph*{Partitioning the Set of Generators.}
In the rest of this section we work with an instance of the Membership
Problem in which the generators are
$A_i=( \boldsymbol a_i,\boldsymbol b_i,c_i)$, for $i=1,\ldots,k$, and the
target matrix is $A=(\boldsymbol a,\boldsymbol b,c)$. Recalling the
homomorphism $\psi:\Heis(n,\mathbb{Z})\rightarrow \mathbb{Z}^{2n-4}$,
let us define
$\boldsymbol v_i := \psi(A_i)=( \boldsymbol a_i,\boldsymbol b_i)$,
for $i=1,\ldots,k$, and
$\boldsymbol v := \psi(A) = (\boldsymbol a,\boldsymbol b)$.

A set $C\subseteq \mathbb{R}^n$ is called a \emph{cone} if
$\sum_{i=1}^k r_i \boldsymbol u_i \in C$ for all
$r_1,\ldots,r_k \in \mathbb{R}_{\geq 0}$ and
$\boldsymbol u_1,\ldots,\boldsymbol u_k \in C$.  The
\emph{dual} of a cone $C\subseteq \mathbb{R}^n$ is the cone defined as
\[ C^* := \{ \boldsymbol x \in \mathbb{R}^n : \boldsymbol x^\top \boldsymbol y \geq 0 
\text{ for all } \boldsymbol y \in C\}. \] We will use the fact that $C= C^{**}$, i.e., a cone is
equal to its double dual \cite[Chapter 2.6.1]{boyd2004}.

We write $\mathrm{Cone}(\boldsymbol v_1,\ldots,\boldsymbol v_k)$ for
the cone generated by the vectors $\boldsymbol v_1,\ldots,\boldsymbol
v_k$.  We now partition the set of generators $\mathcal{G}$ into two
disjoint sets $\mathcal{G}_0,\mathcal{G}_+$, where
\begin{eqnarray*}
\mathcal{G}_0&:=&\left\{ A_i : \forall \boldsymbol u  \in \mathrm{Cone}(\boldsymbol v_1,\ldots,\boldsymbol v_k)^* \quad
\boldsymbol{v}_i^\top \boldsymbol u =0  \right\} \\
\mathcal{G}_+&:=&\left\{ A_i : \exists \boldsymbol u \in \mathrm{Cone}(\boldsymbol v_1,\ldots,\boldsymbol v_k)^*  \quad
\boldsymbol{v}_i^\top \boldsymbol u > 0 \right\} \, .
\end{eqnarray*}
We can determine the sets $\mathcal{G}_0$ and $\mathcal{G}_+$ using linear
programming \cite{Sch86}. Without loss of generality we can assume that
$\mathcal{G}_0=\{A_1,\ldots,A_\ell\}$ for some $\ell\geq 0$.

We show how to compute a bound $\beta>0$ such that for every
sequence $\mathcal{S}=B_1,\ldots,B_m$ of elements of $\mathcal{G}$
whose product is equal to the target matrix $A$, the number of indices
$i$ such that $B_i\in \mathcal{G}_+$ is at most $\beta$. By definition of
$\mathcal{G}_+$, for each $i\in\{1,\ldots,\ell\}$, there
exists $\boldsymbol{u}_i \in
\mathrm{Cone}(\boldsymbol v_1,\ldots,\boldsymbol v_k)^*$ such that
$\boldsymbol v_i^\top \boldsymbol u_i>0$. Since $\boldsymbol v_j^\top
\boldsymbol u_i\geq 0$ for all $j\neq i$, $\mathcal{S}$ contains at
most $\frac{\boldsymbol{v}^\top \boldsymbol u_i}{\boldsymbol v_i^\top
\boldsymbol u_i}$ occurrences of matrix $A_i$ (or no occurrences if
$\boldsymbol{v}^\top \boldsymbol u_i\leq 0$). Thus we may define
$\beta:=\sum_{i} \frac{\boldsymbol{v}^\top \boldsymbol u_i}{\boldsymbol v_i^\top
\boldsymbol u_i}$ where the sum is take over the indices $i=1,\ldots,\ell$ such
that $\boldsymbol{v}^\top \boldsymbol u_i> 0$

We now consider two cases according to whether $\mathcal{G}_0$ is a commutative
set of matrices.

\paragraph*{Case I: $\mathcal{G}_0$ is commutative}

Consider a sequence
$\mathcal{S}=B_1,\ldots,B_m$ of elements of $\mathcal{G}$.  Let
$B_{i_1},\ldots,B_{i_s}$ be the subsequence of $\mathcal{S}$
containing all occurrences of elements of $\mathcal{G}_+$ in
$\mathcal{S}$, where $0=i_0<i_1<\ldots<i_s<i_{s+1}=m+1$.  For
$i\in\{1,\ldots,\ell\}$ and $j\in\{1,\ldots,s+1\}$, write $n_{i,j}$
for the number of occurrences of $A_i\in \mathcal{G}_0$ in the
subsequence of $\mathcal{S}$ lying strictly between $B_{i_{j-1}}$ and
$B_{i_j}$ (where $B_0$ is interpreted as the beginning of $\mathcal{S}$ and $B_{m+1}$ as the end of $\mathcal{S}$).
The idea is to write a formula for $\log(B_1 \cdots B_m)$
that is a linear form in the variables $n_{i,j}$.

Indeed by Equation~\eqref{eq:baker-campbell-hausdorff}, writing $C_{i_j}:=\log(B_{i_j})$
for $j=1,\ldots,s$ and $D_i:=\log(A_i)$ for $i=1,\ldots,\ell$, we have
\begin{align}
\log(B_1\cdots B_m) = \sum_{j=1}^s C_{i_j} &+ \sum_{i=1}^\ell \sum_{j=1}^{s+1} n_{i,j} D_i +
\sum_{1\leq j<j'\leq s} [C_{i_j},C_{i_{j'}}]\notag \\
 &+ \sum_{i=1}^\ell \sum_{1\leq j\leq j'\leq s} n_{i,j}[D_i,C_{i_{j'}}] +
  \sum_{i=1}^\ell \sum_{1\leq j<j'\leq s+1} n_{i,j'}[C_{i_{j}},D_i] 
\label{eq:log-product}
\end{align}

An important observation is that the above formula has no quadratic
terms due to commutativity of $\mathcal{G}_0$.
Now $B_1\cdots B_m=A$ if and only if $\log(B_1 \cdots B_m)=\log(A)$.
Setting the right-hand-side of~\eqref{eq:log-product} equal to $\log(A)$
yields a linear
Diophantine equation in variables $n_{i,j}$.  The form of this
equation is determined by the subsequence of matrices
$B_{i_1},\ldots,B_{i_s}$ lying in $\mathcal{G}_+$.  Recall that we can
without loss of generality restrict attention to the case that $s\leq
\beta$ and thus we reduce the question of whether $A$ lies in the
semigroup generated by $\mathcal{G}$ to the solubility of finitely
many linear equations in nonnegative integers.

\paragraph*{Case II: $\mathcal{G}_0$ is not commutative}
Let $\mathcal{G}_0=\{A_1,\ldots,A_\ell\}$ for some $\ell\geq 2$ such that $A_1$ and $A_2$ do not commute.
Recall that by definition of $\mathcal{G}_0$ it holds that $\boldsymbol v_i^\top \boldsymbol u = 0$ for all $\boldsymbol u  \in
\mathrm{Cone}(\boldsymbol v_1,\ldots,\boldsymbol v_k)^*$ and $i=1,\ldots,\ell$.
Therefore,
\begin{gather}
    \mathrm{Span}(\boldsymbol v_1,\ldots,\boldsymbol v_\ell) \subseteq
    \mathrm{Cone}(\boldsymbol v_1,\ldots,\boldsymbol v_k)^{**} = 
\mathrm{Cone}(\boldsymbol v_1,\ldots,\boldsymbol v_k)  \, .
\label{eq:cone-span}
\end{gather}
Following ideas from \cite{KNP18}, we will show that there exist
integers $p>0$ and $q<0$ such that 
$M_{+}=(\boldsymbol 0,\boldsymbol 0,p)$ and
$M_{-}=(\boldsymbol 0,\boldsymbol 0,q)$ and
 both lie in the semigroup
generated by $\mathcal{G}$.

Indeed, from Equation~\eqref{eq:cone-span} it follows that
$-(\boldsymbol{v}_1+\boldsymbol{v}_2)$ lies in
$\mathrm{Cone}(\boldsymbol v_1,\ldots,\boldsymbol v_k)$.  Thus there
exist $r_1,\ldots,r_k\in \mathbb{R}_{\geq 0}$ with $r_1,r_2$ strictly
  positive such that
  $\sum_{i=1}^k r_i \boldsymbol{v}_i = \boldsymbol{0}$.  But since the
  vectors $\boldsymbol{v}_i$ have integer coefficients we can solve
  the above equation in natural numbers $r_1,\ldots,r_k$ with
  $r_1,r_2>0$.  Taking a sequence of matrices $B_1,\ldots,B_m$, drawn
  from $\mathcal{G}$, such that $B_1=A_1$, $B_2=A_2$ and such that
  matrix $A_i$ appears $r_i$ times in the sequence for
  $i\in\{1,\ldots,k\}$, we obtain $\psi(B_1\cdots B_m)=\boldsymbol 0$.  Since
  $\psi$ is a homomorphism to a commutative group we have that
  $\psi(B_{\sigma(1)}^t \cdots
  B_{\sigma(m)}^t)=\boldsymbol{0}$ for all
  $t\geq 1$ and permutations $\sigma \in S_m$.

Write $C_i=\log(B_i)$ for $i=1,\ldots,m$.  Applying the
Baker-Campbell-Hausdorff Formula~\eqref{eq:baker-campbell-hausdorff},
we have that for any positive integer $t$ and permutation
$\sigma\in S_m$
\begin{eqnarray}
\log(B_{\sigma(1)}^t \cdots B_{\sigma(m)}^t) =
t \sum_{i=1}^m C_{\sigma(i)} + \frac{t^2}{2} \sum_{i<j} [C_{\sigma(i)},C_{\sigma(j)}] .
\label{eq:powers}
\end{eqnarray}

We show that we can obtain the desired matrices
$M_{+}$ and $M_{-}$ as $M_{+} := B_{\sigma(1)}^{t} \cdots
B_{\sigma(m)}^{t}$ and $M_{-} := B_{\sigma(m)}^{t} \cdots B_{\sigma(1)}^{t}$
for some permutation $\sigma \in S_m$ and large enough $t$.

Let $\sigma_0 \in S_m$ be the permutation that transposes $1$ and $2$.
Write also $\mathrm{id} \in S_m$ for the identity permutation.
Defining $\delta_\sigma:=\sum_{i<j} [C_{\sigma(i)},C_{\sigma(j)}]_{1,n}$, we
have $\delta_{\mathrm{id}} - \delta_{\sigma_0} = 2[C_1,C_2]_{1,n} \neq 0$ since
$B_1$, $B_2$ do not commute.
Hence there exists $\sigma\in \{\mathrm{id},\sigma_0\}$ with
$\delta_{\sigma}\neq 0$.  Defining the reverse permutation
$\sigma'\in S_m$ by $\sigma'(i)=\sigma(m+1-i)$ for $i=1,\ldots,m$, we
moreover have $\delta_{\sigma'} = - \delta_\sigma$, and thus we may
suppose that $\delta_\sigma>0$ and $\delta_{\sigma'}<0$.  It remains
to note, by inspection of~\eqref{eq:powers}, that for $t$ sufficiently
large, if $\delta_\sigma\neq 0$ then the sign of the $(1,n)$-entry of
$\log(B_{\sigma(1)}^{t} \cdots B_{\sigma(m)}^{t})$ is equal to the sign of
$\delta_\sigma$.  But since $\log(B_{\sigma(1)}^{t} \cdots B_{\sigma(m)}^{t})$
has zeros in all entries, except for the $(1,n)$-entry, this entry is in fact
equal to the $(1,n)$-entry of $B_{\sigma(1)}^{t} \cdots B_{\sigma(m)}^{t}$.

So, under the assumption that $\mathcal{G}_0$ is not commutative we have
shown that one can compute integers $p>0$ and $q<0$ such that 
$M_{+}=(\boldsymbol 0,\boldsymbol 0,p)$ and
$M_{-}=(\boldsymbol 0,\boldsymbol 0,q)$ are in $\mathcal{G}$. It follows that
$\langle \mathcal{G} \rangle$ contains the group
$\mathcal{N}=\{(\boldsymbol 0,\boldsymbol 0,c) \in \Heis(n,\mathbb{Z})
: c \equiv 0 \pmod m \}$, where $m=\gcd(p,q)$.  
Since
\[
  (\boldsymbol{a},\boldsymbol{b},c)\cdot (\boldsymbol{0},\boldsymbol{0},c') = 
  (\boldsymbol{0},\boldsymbol{0},c')\cdot (\boldsymbol{a},\boldsymbol{b},c) = 
  (\boldsymbol{a},\boldsymbol{b},c+c'),
\]
we have the following equivalence for the target matrix
$A = (\boldsymbol{a},\boldsymbol{b},c)$:
\begin{eqnarray*}
  A = (\boldsymbol{a},\boldsymbol{b},c) \in \langle \mathcal{G} \rangle & \text{ if{}f } &
  \exists B \in \langle \mathcal{G} \rangle \text{ such that } AB^{-1}\in \mathcal{N}\\
  & \text{ if{}f } & \exists B \in \langle \mathcal{G} \rangle \text{ such that }
 B = (\boldsymbol{a},\boldsymbol{b},c') \text{ and }  c'\equiv c \pmod m \, .
\end{eqnarray*}

To decide whether $\langle \mathcal{G} \rangle$ contains a matrix $B = (\boldsymbol{a},
\boldsymbol{b},c')$ with $c'\equiv c \pmod m$, we will use register automata. Let $d=n-2$ and
consider the following finite automaton with $2d$ registers:
\[
  \mathbf{Q} = (\{A_1,\ldots,A_k\},S,R_1,\ldots,R_d,T_1,\ldots,T_d,s_0,\delta,F),
\]
where the alphabet of $\mathbf{Q}$ is equal to the set of generator matrices $\G =
\{A_1,\ldots,A_k\}$, and the set of states $S$ is equal to
\[
  S=\{(s_1,\ldots,s_d, t_1,\ldots,t_d,u)\ :\ s_i,t_i,u\in \{0,\ldots,m-1\} \text{ for }
  i=1,\ldots,d\, \}.
\]
Intuitively, $(2d+1)$-tuples from $S$ store the values of a vector
$(\boldsymbol{a},\boldsymbol{b},c)$ modulo $m$, and the registers $R_1,\ldots,R_d$ and
$T_1,\ldots,T_d$ store the values of $\boldsymbol{a}$ and $\boldsymbol{b}$, respectively.

The initial state of $\mathbf{Q}$ is $s_0 = (0,\ldots,0)$, and the initial values of all the
registers are zeros.  The transition function $\delta$ is defined as follows.  Suppose $\mathbf{Q}$
is in a state $(s_1,\ldots,s_d, t_1,\ldots,t_d,u)$, and the current values of $R_i$ and $T_i$ are
$r_i$ and $t_i$, respectively, for $i=1,\ldots,d$. If $\mathbf{Q}$ reads a letter $A_\ell =
(a^\ell_1,\ldots,a^\ell_d, b^\ell_1,\ldots,b^\ell_d,c^\ell)$, then it moves to the state
$(s'_1,\ldots,s'_d, t'_1,\ldots,t'_d,u')$, where for each $i=1,\ldots,d$:
\begin{align*}
  &s'_i \equiv s_i + a^\ell_i \pmod m \text{ and } t'_i \equiv t_i + b^\ell_i \pmod m,\\
  &u' \equiv u + c^\ell + s_1b^\ell_1 + \cdots + s_db^\ell_d \pmod m.\\
  &\text {Also, the new value of } R_i \text{ is } r_i+a^\ell_i \text { and the new value of } T_i \text{ is } t_i+b^\ell_i.
\end{align*}
The set $F$ of final states consists of one state that corresponds to the values of the target matrix
$A = (\boldsymbol{a},\boldsymbol{b},c)=(a_1,\ldots,a_d,b_1,\ldots,b_d,c)$ modulo $m$, that is
\[
  F=\{ (s_1,\ldots,s_d, t_1,\ldots,t_d,u)\ :\ s_i\equiv a_i,\ t_i\equiv b_i,\ u\equiv c \pmod m \ \text{ for } i=1,\ldots,d\}.
\]
The automaton $\mathbf{Q}$ accepts a word $w\in {\{A_1,\cdots,A_k\}}^*$ if after reading $w$ it
reaches the final state from $F$ and the values of the registers $R_1,\ldots,R_d$ and
$T_1,\ldots,T_d$ are equal to $a_1,\ldots,a_d$ and $b_1,\ldots,b_d$, respectively.  By construction,
the language of $\mathbf{Q}$ in non-empty if and only if $\langle \mathcal{G} \rangle$ contains a
matrix $B = (\boldsymbol{a}, \boldsymbol{b},c')$ with $c'\equiv c \pmod m$.

Note that after reading any letter the registers of $\mathbf{Q}$ are changed by constant values, and
the transitions have no guards or zero checks. Let $\s$ be the set of values that the registers of
$\mathbf{Q}$ can have when it reaches the final state. It is well-known that for a register
automaton of this type the set $\s$ is effectively semilinear (see \cite{KP03,KP02} for details).  In
particular, we can decide whether $\s$ contains the vector $(\boldsymbol{a}, \boldsymbol{b})$, and
so the emptiness problem for $\mathbf{Q}$ is decidable.  Hence, in the case when $\mathcal{G}_0$ is
not commutative the Membership Problem for $\Heis(n,\mathbb{Z})$ is decidable.

\begin{corollary}
The Membership Problem for $\Heis(n,\mathbb{Q})$ is decidable.
\end{corollary}

\begin{proof}
  Let $A_i=( \boldsymbol a_i,\boldsymbol b_i,c_i)$, for $i=1,\ldots,k$, and $A=(\boldsymbol a,
  \boldsymbol b, c)$ be the given generators and the target matrix from $\Heis(n,\mathbb{Q})$. Let
  $N$ be a natural number such that $A_i=( \frac{1}{N}\boldsymbol a'_i,\frac{1}{N}\boldsymbol
  b'_i,\frac{1}{N^2}c'_i)$, for $i=1,\ldots,k$, and $A=(\frac{1}{N}\boldsymbol a' ,\frac{1}{N}
  \boldsymbol b',\frac{1}{N^2}c')$, where $\boldsymbol a'_i, \boldsymbol b'_i, c'_i$, for
  $i=1,\ldots,k$, and $\boldsymbol a', \boldsymbol b',c'$ are integer vectors and numbers.
  It is easy to check that
  \[
    (\tfrac{1}{N}\boldsymbol x,\tfrac{1}{N}\boldsymbol y,\tfrac{1}{N^2}c)\cdot
    (\tfrac{1}{N}\boldsymbol x',\tfrac{1}{N}\boldsymbol y',\tfrac{1}{N^2}c') =
    (\tfrac{1}{N}(\boldsymbol x + \boldsymbol x'),\tfrac{1}{N}(\boldsymbol y + \boldsymbol
    y'),\tfrac{1}{N^2} (c+c' + \boldsymbol x^\top\boldsymbol y'))\,.
  \]
  Hence $A\in \la A_1,\ldots,A_k\ra$ if{}f $A'\in \la A'_1,\ldots,A'_k\ra$, where $A'=(\boldsymbol a',
  \boldsymbol b', c')$ and $A'_i=( \boldsymbol a'_i,\boldsymbol b'_i,c'_i)$, for $i=1,\ldots,k$, are
  matrices with integer entries, that is, from $\Heis(n,\mathbb{Z})$. By Theorem \ref{thm:heis} we can
  decide whether $A'\in \la A'_1,\ldots,A'_k\ra$.
\end{proof}

\section{The Half-Space Reachability Problem for $\GL$}\label{sec:gl}

In this section we will show that the Half-Space Reachability Problem for $\GL$ is decidable
(Theorem \ref{thm:halfpl}).

\begin{definition}
  For an integer $n$, the sign of $n$ as follows: $\sg(n)=1$ if $n>0$,
  $\sg(n)=-1$ if $n<0$, and $\sg(n)=*$ if $n=0$.

For a matrix $A=\begin{pmatrix}a & b\\ c & d\end{pmatrix}\in \Z^{2\times 2}$,
define $\sg(A):=\begin{pmatrix} \sg(a) & \sg(b)\\ \sg(c) & \sg(d)\end{pmatrix}$.
\end{definition}

If $A$ and $B$ are two expressions whose values are in the set $\{1,-1,*\}$,
then the notation $A\simeq B$ means that $A=B$ or $A=*$ or $B=*$.

\begin{proposition}\label{pr:1}
Suppose $w$ is a canonical word of the form $w=SR^{\alpha_1}SR^{\alpha_2}\cdots
SR^{\alpha_n}$, where $\alpha_i\in \{1,2\}$ for $i=1,\ldots,n$. Then
$\sg(\phi(w))\simeq \begin{pmatrix} {(-1)}^n & {(-1)}^n\\ {(-1)}^n & {(-1)}^n\end{pmatrix}$.
\end{proposition}

\begin{proof}
The proof is by induction on $n$. For $n=1$, we have
\begin{align*}
\sg(\phi(SR)) &= \sg \begin{pmatrix}-1 & -1\\ 0 & -1\end{pmatrix} =
\begin{pmatrix}\sg(-1) & \sg(-1)\\ \sg(0) & \sg(-1)\end{pmatrix} \simeq \begin{pmatrix}-1 & -1\\ -1 & -1\end{pmatrix}\quad \text{and}\\
\sg(\phi(SR^2)) &= \sg \begin{pmatrix}-1 & 0\\ -1 & -1\end{pmatrix}  =
\begin{pmatrix}\sg(-1) & \sg(0)\\ \sg(-1) & \sg(-1)\end{pmatrix} \simeq \begin{pmatrix}-1 & -1\\ -1 & -1\end{pmatrix}
\end{align*}

Suppose the statement of the proposition is true for $w=SR^{\alpha_1}SR^{\alpha_2}\cdots SR^{\alpha_n}$ and consider the words $\mathit{wSR}$ and $\mathit{wSR}^2$. Assume that $\phi(w)=\begin{pmatrix}a & b\\ c & d\end{pmatrix}$ and
\[
  \sg(\phi(w))= \begin{pmatrix}\sg(a) & \sg(b)\\ \sg(c) & \sg(d)\end{pmatrix}
  \simeq \begin{pmatrix} {(-1)}^n & {(-1)}^n\\ {(-1)}^n & {(-1)}^n\end{pmatrix}\!.
\]
Then we have
\begin{align*}
\phi(\mathit{wSR}) = \phi(w)\phi(\mathit{SR}) &= \begin{pmatrix}a & b\\ c & d\end{pmatrix}\begin{pmatrix}-1 & -1\\ 0 & -1\end{pmatrix} = \begin{pmatrix}-a & -a-b\\ -c & -c-d\end{pmatrix}\quad \text{and}\\
\phi(\mathit{wSR}^2) = \phi(w)\phi(\mathit{SR}^2) &= \begin{pmatrix}a & b\\ c & d\end{pmatrix}\begin{pmatrix}-1 & 0\\ -1 & -1\end{pmatrix} = \begin{pmatrix}-a-b & -b\\ -c-d & -d\end{pmatrix}\!.
\end{align*}
From these formulas it not hard to see that
\[
  \sg(\phi(\mathit{wSR})) \simeq \begin{pmatrix} {(-1)}^{n+1} & {(-1)}^{n+1}\\
  {(-1)}^{n+1} & {(-1)}^{n+1}\end{pmatrix}\quad \text{and}\quad
  \sg(\phi(\mathit{wSR}^2)) \simeq \begin{pmatrix} {(-1)}^{n+1} & {(-1)}^{n+1}\\
  {(-1)}^{n+1} & {(-1)}^{n+1}\end{pmatrix}.
\]
\end{proof}

\begin{proposition}\label{pr:2}
Let $w$ be a canonical word of the form $w=S^\beta
R^{\alpha_1}SR^{\alpha_2}\cdots SR^{\alpha_n}S^\epsilon$, where
$\beta,\epsilon\in \{0,1\}$ and $\alpha_i\in \{1,2\}$, $i=1,\ldots,n$. Then
$
  \sg(\phi(w))\simeq\begin{pmatrix} {(-1)}^n & {(-1)}^{n+\epsilon}\\ {(-1)}^{n-1+\beta} & {(-1)}^{n-1+\beta+\epsilon}\end{pmatrix}
$.
\end{proposition}

\begin{proof}
  First, consider the case when $\epsilon=0$. Suppose $\phi(SR^{\alpha_1}SR^{\alpha_2}\cdots SR^{\alpha_n}) = \begin{pmatrix}a & b\\ c & d\end{pmatrix}$. Then by Proposition \ref{pr:1} we have
  \[
    \sg(\phi(SR^{\alpha_1}SR^{\alpha_2}\cdots SR^{\alpha_n})) =
    \begin{pmatrix}\sg(a) & \sg(b)\\ \sg(c) & \sg(d)\end{pmatrix} \simeq \begin{pmatrix} {(-1)}^n & {(-1)}^n\\ {(-1)}^n & {(-1)}^n\end{pmatrix}.
  \]
  On the other hand,
  \begin{align*}
    \phi(R^{\alpha_1}SR^{\alpha_2}\cdots SR^{\alpha_n}) &= -\phi(S)\phi(SR^{\alpha_1}SR^{\alpha_2}\cdots SR^{\alpha_n})\\
    &= \begin{pmatrix}0 & 1\\ -1 & 0\end{pmatrix} \begin{pmatrix}a & b\\ c & d\end{pmatrix} = \begin{pmatrix}c & d\\ -a & -b\end{pmatrix}\!.
  \end{align*}
  Hence $\sg(\phi(R^{\alpha_1}SR^{\alpha_2}\cdots SR^{\alpha_n})) =
  \begin{pmatrix}\sg(c) & \sg(d)\\ \sg(-a) & \sg(-b)\end{pmatrix} \simeq \begin{pmatrix} {(-1)}^n & {(-1)}^n\\ {(-1)}^{n-1} & {(-1)}^{n-1}\end{pmatrix}$.
  Thus, for $\beta\in \{0,1\}$, we showed that 
  \begin{equation}
    \label{eq:1}
    \sg(\phi(S^\beta R^{\alpha_1}SR^{\alpha_2}\cdots SR^{\alpha_n}))\simeq\begin{pmatrix} {(-1)}^n & {(-1)}^{n}\\ {(-1)}^{n-1+\beta} & {(-1)}^{n-1+\beta}\end{pmatrix}\!.
  \end{equation}

  Now we consider the case when $\epsilon=1$. Suppose $\phi(S^\beta R^{\alpha_1}SR^{\alpha_2}\cdots SR^{\alpha_n}) = \begin{pmatrix}a & b\\ c & d\end{pmatrix}$. Then
  \begin{equation}\label{eq:2}
    \begin{aligned}
      \phi(S^\beta R^{\alpha_1}SR^{\alpha_2}\cdots SR^{\alpha_n}S) &= \phi(S^\beta R^{\alpha_1}SR^{\alpha_2}\cdots SR^{\alpha_n})\phi(S)\\
      &= \begin{pmatrix}a & b\\ c & d\end{pmatrix}\begin{pmatrix}0 & -1\\ 1 & 0\end{pmatrix} = \begin{pmatrix}b & -a\\ d & -c\end{pmatrix}\!.
    \end{aligned}
  \end{equation}
  From equations (\ref{eq:1}) and (\ref{eq:2}) we obtain
  \begin{equation}\label{eq:3}
    \sg(\phi(S^\beta R^{\alpha_1}SR^{\alpha_2}\cdots SR^{\alpha_n}S)) =
    \begin{pmatrix}\sg(b) & \sg(-a)\\ \sg(d) & \sg(-c)\end{pmatrix} \simeq \begin{pmatrix} {(-1)}^n & {(-1)}^{n+1}\\ {(-1)}^{n-1+\beta} & {(-1)}^{n-1+\beta+1}\end{pmatrix}\!.
  \end{equation}
  Equations (\ref{eq:1}) and (\ref{eq:3}) imply that for $\beta,\epsilon\in \{0,1\}$
  \[
    \sg(\phi(S^\beta R^{\alpha_1}SR^{\alpha_2}\cdots SR^{\alpha_n}S^\epsilon))
    \simeq \begin{pmatrix} {(-1)}^n & {(-1)}^{n+\epsilon}\\ {(-1)}^{n-1+\beta} & {(-1)}^{n-1+\beta+\epsilon}\end{pmatrix}\!.
  \]
\end{proof}

From Proposition \ref{pr:2} and the equalities
\[
  \phi(X)\begin{pmatrix}a & b\\ c & d\end{pmatrix} = \begin{pmatrix}-a & -b\\ -c & -d\end{pmatrix}\quad \text{and}\quad
  \phi(N)\begin{pmatrix}a & b\\ c & d\end{pmatrix} = \begin{pmatrix}1 & 0\\ 0 & -1\end{pmatrix} \begin{pmatrix}a & b\\ c & d\end{pmatrix} = \begin{pmatrix}a & b\\ -c & -d\end{pmatrix}
\]
we obtain the following proposition.

\begin{proposition}\label{pr:3}
  Let $w$ be a canonical word of the form $w=N^\delta X^\gamma S^\beta R^{\alpha_1}SR^{\alpha_2}\cdots SR^{\alpha_n}S^\epsilon$, where $\beta,\gamma, \delta, \epsilon\in \{0,1\}$ and $\alpha_i\in \{1,2\}$ for $i=1,\ldots,n$. Then
  \[
    \sg(\phi(w))\simeq\begin{pmatrix} {(-1)}^{n+\gamma} & {(-1)}^{n+\gamma+\epsilon}\\ {(-1)}^{n-1+\beta+\gamma+\delta} & {(-1)}^{n-1+\beta+\gamma+\delta+\epsilon}\end{pmatrix}\!.
  \]
\end{proposition}

\begin{theorem} \label{thm:pos}
  The set of matrices in $\GL$ whose particular entry is nonnegative forms a regular subset. In other words, for all $i,j\in \{1,2\}$, the following subset of $\GL$ is regular:
  \[
    \mathrm{Pos}_{ij} = \left\{\begin{pmatrix}a_{11} & a_{12}\\a_{21}  & a_{22}\end{pmatrix}\in \GL\ :\ a_{ij}\geq 0 \right\}.
  \]
\end{theorem}

\begin{proof}
  Suppose $i=j=2$ as other cases are similar. Let $A$ be a matrix from $\GL$ and let
  \[
    w=N^\delta X^\gamma S^\beta R^{\alpha_1}SR^{\alpha_2}\cdots SR^{\alpha_n}S^\epsilon,
  \]
  where $\beta,\gamma, \delta, \epsilon\in \{0,1\}$ and $\alpha_i\in \{1,2\}$
  for $i=1,\ldots,n$, be a canonical word that represents $A$, that is,
  $A=\phi(w)$. From Proposition \ref{pr:3} we see that $\sg(a_{22}) \simeq
  {(-1)}^{n-1+\beta+\gamma+\delta+\epsilon}$. Hence
  \begin{equation}
    \label{eq:eqv}
    a_{22}\geq 0\quad \text{if and only if}\quad n-1+\beta+\gamma+\delta+\epsilon \equiv 0 \pmod 2.
  \end{equation}
  
  To finish the proof, we note that the set of all canonical words is regular. Furthermore, given a canonical word of the form $w=N^\delta X^\gamma S^\beta R^{\alpha_1}SR^{\alpha_2}\cdots SR^{\alpha_n}S^\epsilon$, a finite automaton can read off the values of $\beta,\gamma,\delta,\epsilon$ and determine the parity of number $n$. From this data an automaton can decide whether $a_{22}\geq 0$ by the above mentioned equivalence (\ref{eq:eqv}). Hence the set of canonical words $w$ such that $\phi(w)\in \mathrm{Pos}_{22}$ can be recognised by a finite automaton.
\end{proof}

Next theorem was proved in \cite{PS17}.

\begin{theorem} \label{thm:val}
  For every $k\in\Z$, the following subset of $\GL$ is regular:
  \[
    \mathrm{S}_{ij}(k) = \left\{\begin{pmatrix}a_{11} & a_{12}\\a_{21}  & a_{22}\end{pmatrix}\in \GL\ :\ a_{ij} = k \right\}.
  \]
\end{theorem}

As a corollary from Theorems \ref{thm:pos} and \ref{thm:val} we obtain:

\begin{theorem} \label{thm:bound}
  For every $k\in\Z$, the following subsets of $\GL$ are regular:
  \begin{align*}
    \mathrm{S}_{ij}(\geq\! k) &= \left\{\begin{pmatrix}a_{11} & a_{12}\\a_{21}  & a_{22}\end{pmatrix}\in \GL\ :\ a_{ij}\geq k \right\}\quad\text{and}\\
    \mathrm{S}_{ij}(\leq\! k) &= \left\{\begin{pmatrix}a_{11} & a_{12}\\a_{21}  & a_{22}\end{pmatrix}\in \GL\ :\ a_{ij}\leq k \right\}.
  \end{align*}
\end{theorem}

\begin{proof}
  Since $\mathrm{S}_{ij}(\leq\! k)$ is the complement of $\mathrm{S}_{ij}(\geq\! k+1)$, it suffices to prove that the sets $\mathrm{S}_{ij}(\geq\! k)$ are regular.

  If $k=0$, then it follows from Theorem \ref{thm:pos} that $\mathrm{S}_{ij}(\geq\! 0) = \mathrm{Pos}_{ij}$ is regular. Furthermore,
  \[
    \mathrm{S}_{ij}(\geq\! k) = \mathrm{Pos}_{ij} \ \setminus\
    \bigcup_{n=0}^{k-1} M_{ij}(n) \text{ if } k>0\quad \text{and}\quad
    \mathrm{S}_{ij}(\geq\! k) = \mathrm{Pos}_{ij}\ \cup\ \bigcup_{n=k}^{-1} M_{ij}(n) \text{ if } k<0.
  \]
  Since by Corollary \ref{thm:bool} regular subsets of $\GL$ are closed under Boolean operations, we conclude that $\mathrm{S}_{ij}(\geq\! k)$ is a regular set for any $k\in \Z$.
\end{proof}

\begin{theorem} \label{thm:half}
  Let $\lambda\in\Q$ and $\boldsymbol{u},\boldsymbol{v}\in \Q\times\Q$. Then the set
  $
    \s(\boldsymbol{u},\boldsymbol{v},\lambda) = \{\, M\in \GL\ :\ \boldsymbol{u}^\top M\boldsymbol{v}\geq \lambda\,\}
  $
  is a regular subset of $\GL$.
\end{theorem}

\begin{proof}
  Note that if $\boldsymbol{u}=\boldsymbol{0}$ or $\boldsymbol{v}=\boldsymbol{0}$, then $\boldsymbol{u}^\top M\boldsymbol{v}=0$. In this case $\s(\boldsymbol{u},\boldsymbol{v},\lambda)$ equals either the empty set or $\GL$, both of which are regular subsets.
  Hence we will assume that both $\boldsymbol{u} = \begin{pmatrix}u_1\\ u_2\end{pmatrix}$ and $\boldsymbol{v}=\begin{pmatrix}v_1\\ v_2\end{pmatrix}$ are nonzero vectors. By multiplying the inequality $\boldsymbol{u}^\top M\boldsymbol{v}\geq \lambda$ by the least common multiple of the denominators of $u_1,u_2,v_1,v_2$, we can assume that $\boldsymbol{u}$ and $\boldsymbol{v}$ have integer coefficients. Furthermore, we can divide $\boldsymbol{u}^\top M\boldsymbol{v}\geq \lambda$ by $\gcd(u_1,u_2)$ and $\gcd(v_1,v_2)$ and so assume from now on that $\gcd(u_1,u_2)=\gcd(v_1,v_2)=1$.

  Finally, note that the inequality $\boldsymbol{u}^\top M\boldsymbol{v}\geq \lambda$ is equivalent to $\boldsymbol{u}^\top M\boldsymbol{v}\geq \lceil\lambda\rceil$, where $\lceil\lambda\rceil = \min \{n\in \Z\ :\ n\geq \lambda\}$. So, we can assume that $\lambda$ is also an integer number.

  Since $\gcd(u_1,u_2)=\gcd(v_1,v_2)=1$, there are integers $s_1$, $s_2$, $t_1$, $t_2$ such that $s_1u_1+s_2u_2=1$ and $t_1v_1+t_2v_2=1$. Hence the matrices $A=\begin{pmatrix} u_1 & -s_2\\ u_2 & s_1 \end{pmatrix}$ and $B=\begin{pmatrix} v_1 & -t_2\\ v_2 & t_1 \end{pmatrix}$ belong to $\GL$, and we have that $\boldsymbol{u}=A\boldsymbol{e}_1$ and $\boldsymbol{v}=B\boldsymbol{e}_1$. Therefore, the inequality $\boldsymbol{u}^\top M\boldsymbol{v} \geq \lambda$ is equivalent to $\boldsymbol{e}_1^\top A^\top MB\boldsymbol{e}_1 \geq \lambda$. In other words,
  \[
    M\in \s(\boldsymbol{u},\boldsymbol{v},\lambda)\quad \Longleftrightarrow \quad A^\top MB \in \mathrm{S}_{11}(\geq\! \lambda)\quad \Longleftrightarrow \quad M \in {(A^\top)}^{-1}\cdot \mathrm{S}_{11}(\geq\! \lambda)\cdot B^{-1}.
  \]
  By Theorem \ref{thm:bound}, $\mathrm{S}_{11}(\geq\! \lambda)$ is a regular subset of $\GL$. Let $L$ be a regular language and let $w_1$, $w_2$ be canonical words such that $\phi(L)=\mathrm{S}_{11}(\geq\! \lambda)$ and $\phi(w_1)={(A^\top)}^{-1}$ and $\phi(w_2)=B^{-1}$. Then $\{w_1\}\cdot L \cdot \{w_2\}$ is a regular language such that
  \[
    \phi(\{w_1\}\cdot L \cdot \{w_2\}) = {(A^\top)}^{-1}\cdot \mathrm{S}_{11}(\geq\! \lambda)\cdot B^{-1} = \s(\boldsymbol{u},\boldsymbol{v},\lambda).
  \]
\end{proof}

\begin{theorem} \label{thm:halfpl}
  The Half-Space Reachability Problem for $\GL$ is decidable.
\end{theorem}

\begin{proof}
  Let $\G=\{ A_1,\dots,A_k\}$ be a finite collection of matrices from $\GL$, $\lambda$ be a rational
  number and $\boldsymbol{u},\boldsymbol{v}$ be vectors from $\Q^2$.
  Define $\s(\boldsymbol{u},\boldsymbol{v},\lambda) := \{\, M\in \GL\ :\ \boldsymbol{u}^\top M\boldsymbol{v}\geq \lambda\,\}$. By Theorem \ref{thm:half}, $\s(\boldsymbol{u},\boldsymbol{v},\lambda)$ is a regular subset of $\GL$. Let $L_\s$ be a regular language such that $\s(\boldsymbol{u},\boldsymbol{v},\lambda) = \phi(L_\s)$. It is not hard to see that the semigroup $\la\G\ra$ is also a regular subset. Indeed, consider a regular language $L_\G={(w_1\cup\cdots\cup w_k)}^+$, where $w_1,\dots,w_k$ are canonical words that correspond to the matrices $A_1,\dots,A_k$, respectively. Then $\la\G\ra=\phi(L_\G)$.

  By Corollary \ref{thm:bool}, we can algorithmically construct a regular language $L^\cap$ such that
  \[
    \phi(L^\cap) = \phi(L_\s) \cap \phi(L_\G) = \s(\boldsymbol{u},\boldsymbol{v},\lambda) \cap \la\G\ra.
  \]
  Now we have the following equivalence:
  \[
    \text{there is } M\in \la\G\ra \text{ such that } \boldsymbol{u}^\top M\boldsymbol{v}\geq \lambda
   \quad \text{if{}f} \quad
    \s(\boldsymbol{u},\boldsymbol{v},\lambda) \cap \la\G\ra = \phi(L^\cap)\neq \emptyset.
  \]
  The last condition is equivalent to $L^\cap\neq\emptyset$. Therefore, we reduced the Half-Space
  Reachability Problem for $\GL$ to the emptiness problem for regular languages.
\end{proof}

\section{The Half-Space Reachability Problem for the Heisenberg
Group}\label{sec:half}

\begin{definition}
Let $\mathcal{S}:=B_1,\ldots,B_m$ be a sequence in $\Heis(n,\mathbb{Q})$ and $A$ a particular matrix
in $\Heis(n,\mathbb{Q})$.
A pair $i,j \in \{1\ldots,m\}$ with $i\leq j$
is called an $A$-\emph{block} of $\mathcal{S}$ if
\begin{enumerate}
\item $B_k=A$ for all $k\in\{i,\ldots,j\}$,
\item either $i=1$ or $B_{i-1} \neq A$,
\item either $j=m$ or $B_{j+1} \neq A$.
\end{enumerate}
We say that $\mathcal{S}$ is \emph{pure} if it has at most one $A$-block for every matrix $A$.
\end{definition}

Given a sequence $\mathcal{S}=B_1,\ldots,B_m \in \Heis(n,\mathbb{Q})$,
define $C_i:=\log(B_i)$ for $i=1,\ldots,m$,
$\Delta(\mathcal{S}) := \sum_{1\leq i<j\leq m} [C_i,C_j]$, and
$\delta(\mathcal{S}) := \Delta(\mathcal{S})_{1,n}$.  Recall that using the
Baker-Campbell-Hausdorff formula \eqref{eq:baker-campbell-hausdorff} we can express the
product of the sequence $\mathcal{S}$ as follows
\begin{gather}
  B_1 \cdots B_m = \exp\Big(\sum_{i=1}^m C_i + \frac{1}{2}
  \underbrace{\sum_{1\leq i<j\leq m}
    [C_i,C_j]}_{\Delta(\mathcal{S})}\Big)
\label{eq:baker-campbell-hausdorff-new}
\end{gather}

\begin{proposition}
For any sequence of matrices $\mathcal{S}=B_1,\ldots,B_m \in
\Heis(n,\mathbb{Q})$, there is a permutation $\pi\in S_m$ such that
sequence $\mathcal{S}':=B_{\pi(1)},\ldots,B_{\pi(m)}$ is pure and 
$\delta(\mathcal{S}) \leq \delta(\mathcal{S}')$.
\label{prop:blocks}
\end{proposition}

\begin{proof}
We show that if $\mathcal{S}$ has at least two $A$-blocks for some $A$
then $\mathcal{S}$ can be permuted to to obtain a new sequence
$\mathcal{S}'$ such that $\mathcal{S}'$ has one fewer $A$-block than
$\mathcal{S}$, $\mathcal{S}'$ has at most as many $B$-blocks as
$\mathcal{S}$ for any $B\neq A$, and $\delta(\mathcal{S}) \leq
\delta(\mathcal{S}')$.

Fix some matrix $A \in \Heis(n,\mathbb{Q})$.
Let $(i_1,j_1)$ and $(i_2,j_2)$ be two distinct $A$-blocks in the
sequence $\mathcal{S}$, with $j_1+1<i_2$.  Then we can write
$\mathcal{S}$ in the form
\[ \mathcal{S}=B_1,\ldots,B_{i_1-1},\underbrace{B_{i_1},\ldots,B_{j_1}}_{\text{$A$-block}},B_{j_1+1},\ldots,B_{i_2-1},
\underbrace{B_{i_2},\ldots,B_{j_2}}_{\text{$A$-block}},B_{j_2+1},\ldots,B_m \, . \]

Write $C_i:=\log(B_i)$ for $i=1,\ldots,m$ and $C:=\log(A)$.
We now consider two cases according to the sign of
$\sum_{i=j_1+1}^{i_2-1}[C_i,C]_{1,n}$.

First, suppose that $\sum_{i=j_1+1}^{i_2-1}[C_i,C]_{1,n} \leq 0$.
Then we define
\[
\mathcal{S}':=
B_1,\ldots,B_{i_1-1},\underbrace{B_{i_1},\ldots,B_{j_1},B_{i_2},\ldots,B_{j_2}}_{\text{$A$-block}},
B_{j_1+1},\ldots,B_{i_2-1},
B_{j_2+1},\ldots,B_m
\]
to be the sequence obtained from $\mathcal{S}$ by swapping the order
of the $A$-block $B_{i_2},\ldots,B_{j_2}$ and the preceding
subsequence $B_{j_1+1},\ldots,B_{i_2-1}$.  Notice that $\mathcal{S}'$
has one fewer $A$-block than $\mathcal{S}$ and no more $B$-blocks for
any $B\neq A$.  Moreover, 
\[
\delta(\mathcal{S}) - \delta(\mathcal{S}') =
2(j_2-i_2+1)\sum_{i=j_1+1}^{i_2-1}[C_i,C]_{1,n} \leq 0 \quad \text{and hence}
\quad \delta(\mathcal{S}) \leq \delta(\mathcal{S}').
\]

In the case that $\sum_{i=j_1+1}^{i_2-1}[C_i,C]_{1,n} > 0$ we define
$\mathcal{S}'$ by swapping order of the $A$-block $B_{i_1},\ldots,B_{j_1}$
with the following sequence $B_{j_1+1},\ldots,B_{i_2-1}$.  Then we have
\[
\delta(\mathcal{S}) - \delta(\mathcal{S}') =
-2(j_1-i_1+1)\sum_{i=j_1+1}^{i_2-1}[C_i,C]_{1,n} < 0 \quad \text{and hence}
\quad \delta(\mathcal{S}) < \delta(\mathcal{S}').
\]
\end{proof}

\begin{theorem}\label{thm:heis_half}
The Half-Space Reachability Problem for $\Heis(n,\mathbb{Q})$ is decidable.
\end{theorem}
\begin{proof}
Consider an instance of the Half-Space Reachability Problem, 
given by a finite set $\mathcal{G}=\{A_1,\ldots,A_k\} \subseteq
\Heis(n,\mathbb{Q})$ of generators, vectors $\boldsymbol u,\boldsymbol
v \in \mathbb{Q}^n$ and a scalar $\lambda\in \Q$.

Given a sequence $\mathcal{S}=B_1,\ldots,B_m$ of elements of
$\mathcal{G}$ and a permutation $\sigma\in \mathrm{Sym}_m$, define
$\mathcal{S}_\sigma = B_{\sigma(1)},\ldots,B_{\sigma(m)}$.  It follows from
Equation~\eqref{eq:baker-campbell-hausdorff-new} that the entries of the product
$B_{\sigma(1)}\cdots B_{\sigma(m)}$ do not depend on the choice of $\sigma\in
\mathrm{Sym}_m$, except for the $(1,n)$-entry which is equal to
$\frac{1}{2}\Delta(\mathcal{S}_\sigma)_{1,n}$ plus a constant that also does not
depend on $\sigma$. So, the permutation $\sigma$ that maximises
$\boldsymbol u^\top B_{\sigma(1)}\cdots B_{\sigma(m)} \boldsymbol{v}$
is the same which maximises or minimises $\Delta(\mathcal{S}_\sigma)_{1,n}$
depending on the sign of the coefficient at $\Delta(\mathcal{S}_\sigma)_{1,n}$
in the expression $\boldsymbol{u}^\top \Delta(\mathcal{S}_\sigma)
\boldsymbol{v}$, namely, on the sign of $\boldsymbol{u}_1\boldsymbol{v}_n$. By
Proposition~\ref{prop:blocks} we may assume without loss of generality that the
optimal permutation $\sigma$ is such that $\mathcal{S}_\sigma$ is pure.

By the reasoning above, to decide the given instance of the Half-Space
Reachability Problem it suffices to restrict attention to pure
sequences of generators.  Equivalently we must decide whether there
exist nonnegative integers $n_1,\ldots,n_k$ and a permutation
$\sigma \in \mathrm{Sym}_k$ such that
$\boldsymbol u^\top A^{n_1}_{\sigma(1)} \cdots A^{n_k}_{\sigma(k)}
\boldsymbol{v} \geq \lambda$.  Write $C_i=\log A_i$ for $i=1,\ldots,k$.  Then
\begin{eqnarray*}
\boldsymbol u^\top A^{n_1}_{\sigma(1)} \cdots A^{n_k}_{\sigma(k)} \boldsymbol{v}&=&
\boldsymbol u^\top \exp \left( \sum_{i=1}^k n_i C_{\sigma(i)} +
\frac{1}{2} \sum_{i<j} n_in_j [C_{\sigma(i)},C_{\sigma(j)}] \right) \boldsymbol v \\
&=& Q(n_1,\ldots,n_k) 
\end{eqnarray*}
for some quadratic polynomial $Q(x_1,\ldots,x_k)$ with rational coefficients.

In the work of Grunewald and Segal \cite{SG04} an algorithm is given for solving the following
problem: does there exist integers $n_1,\ldots,n_k$ that satisfy a given quadratic equation
$Q(n_1,\ldots,n_k)=0$ (with rational coefficients) and a finite number of linear inequalities on
$n_1,\ldots,n_k$ (also with rational coefficients).

By introducing a ``dummy'' variable we can use the Grunewald and Segal algorithm to decide whether
$Q(n_1,\ldots,n_k)\geq \lambda$ for some nonnegative integers $n_1,\ldots,n_k$. Hence the Half-Space
Reachability Problem for $\Heis(n,\mathbb{Q})$ is decidable.
\end{proof}

\bibliography{refs}

\end{document}